\documentclass[11pt]{article}

 \usepackage[margin=1in]{geometry}
 \usepackage{tabularx}
 \usepackage{booktabs}

 \usepackage[usenames, dvipsnames]{xcolor}
 \usepackage[colorlinks=true, citecolor=ForestGreen, linkcolor=ForestGreen, urlcolor=NavyBlue]{hyperref}

\usepackage{amsthm,amssymb,amsmath}
\usepackage{bbding}
\usepackage[nointegrals]{wasysym}
\usepackage{xspace}
\usepackage{mathtools}
\usepackage{url}
\usepackage{subcaption}
\usepackage{tikz}
\usetikzlibrary{arrows, fit, positioning}
\usetikzlibrary{backgrounds,automata,calc}
\usetikzlibrary{decorations.pathreplacing}
\usepackage{csquotes}
\usepackage{amsfonts}
\usepackage{bm,bbm}
\usepackage{pgfplots}
\pgfplotsset{compat=1.17}

\usepackage{enumerate}
\usepackage{paralist}
\usepackage[shortlabels]{enumitem}
\usepackage{appendix}
\usepackage{thmtools}
\usepackage{mathtools}
\usepackage{natbib}
\usepackage{changepage}
 \usepackage{authblk}
\usepackage{commath}
\usepackage[ruled,vlined]{algorithm2e}

\usepackage[nameinlink, capitalise]{cleveref}

\usepackage{mathtools}

\setlist{topsep=0.5ex,itemsep=0.1ex}

\newcommand{\vare}{\varepsilon}
\newcommand{\E}{\mathbb{E}}
\renewcommand{\deg}{\textup{\texttt{deg}}}

\newtheorem{openquestion}{Open Question}
\newtheorem{theorem}{Theorem}[section]
\newtheorem{lemma}[theorem]{Lemma}

\theoremstyle{definition}

\newtheorem{definition}[theorem]{Definition}

\usepackage{color-edits}

\addauthor{VV}{red}
\addauthor{RS}{blue}

\title{Improved Fair Graph Partitions}

\author[1]{Nisarg Shah}
\author[2]{Vignesh Viswanathan}
\affil[1]{University of Toronto, \texttt{{nisarg@cs.toronto.edu}}}
\affil[2]{University of Massachusetts Amherst, \texttt{{vviswanathan@umass.edu}}}
\date{}
\begin{document}

\maketitle

\begin{abstract}
We consider the problem of partitioning an undirected graph (representing a social network) with $n$ nodes and maximum degree $\Delta$ into $k$ (almost) equally-sized parts. Each node in the graph, representing an agent, derives utility proportional to the number of their neighbors in their assigned part. Our goal is to find a balanced partition that is {\em fair}. The two notions of fairness we consider are the core and envy-freeness. A partition is envy-free if no node gains utility from moving to a different part, and a partition is in the core if no set of $n/k$ nodes can deviate to form a new part with all its members gaining in utility.

We show that there is a balanced partition which is both $\widetilde{O}\left(1 + \sqrt{\Delta/k}\right)$-approximately envy-free and in the $(k+\sqrt{k},\widetilde{O}(k^{3/2}))$-approximate core. Moreover, such a partition can be computed efficiently. The envy-freeness approximation is a strict improvement over the bound of $\widetilde{O}\left(\sqrt{n/k}\right)$ obtained by \cite{Li2023GraphPartitioning} without any simultaneous approximate core guarantee. Ours is the first result that proves the existence (and computability) of partitions of general graphs which are simultaneously approximately envy-free and in the approximate core. Additionally, for $k = 2$, we prove that every graph with at least $240$ nodes admits a balanced partition in the $(38/25,0)$-core, and a balanced partition in the $(2,0)$-core can be computed efficiently. The last two results make progress on one open question and resolve another posed by \cite{Li2023GraphPartitioning}.
\end{abstract}

\section{Introduction}
We have a set of $n$ agents that need to be divided into groups of equal sizes. This simple statement models several natural problems such as assigning tables to participants at a large banquet event \citep{Deligkas2025Balanced}, or dividing a set of students into teams for a friendly sporting competition \citep{Agarwal2025Harmonious}. The trivial solution to this problem is to arbitrarily create a partition of agents into groups. However, central planners often have an incentive to take interpersonal relationships into account when constructing these partitions, primarily to ensure all the agents have a good experience. This raises the important question of how one can do so while being fair to the agents involved.

There are several models which capture this problem; the specific model we consider was introduced recently by \cite{Li2023GraphPartitioning}. In this natural model, friendships between agents are captured by an undirected graph over the set of agents. An edge exists between two agents if they are friends; friendships are binary and symmetric. The goal is to divide the agents into $k$ groups of equal size (or almost equal size) such that the partition is fair to the agents involved.

The fairness of a partition is measured using two criteria: envy-freeness and the core. A partition is said to be envy-free if no agent is willing to swap their assigned group with any other agent. A partition is in the core if there is no set of $\lfloor n/k \rfloor$ agents who all prefer being grouped with each other than their assigned groups.

\cite{Li2023GraphPartitioning} and subsequent follow-up work \citep{Agarwal2025Harmonious} present several algorithms to compute partitions which are either approximately envy-free or in the approximate core. However, none of these papers consider achieving these two desirable criteria simultaneously for general graphs\footnote{\cite{Agarwal2025Harmonious} show that both criteria can be achieved for the restricted class of grid graphs, and \cite{Li2023GraphPartitioning} show that both criteria can be achieved for the restricted class of trees.}. In this paper, we address this gap by answering the following question:
\begin{center}
    {\em Are there any (approximately) envy-free partitions in the (approximate) core?}
\end{center}
\subsection{Our Contribution}
Before we present our main result, we set up some basic notation. The utility of an agent $i$ in a partition $X$ (denoted $u_i(X)$) is defined as the number of its friends (or neighbors) in its assigned part. A partition is {\em balanced} if all parts have the same size (up to a difference of one). A partition is $r$-approximately envy-free if each agent gains a utility of at most $r$ by swapping its part assignment with another agent. A partition $X$ is in the $(\alpha, \beta)$-core if there is no set of agents $S$ of size $\lfloor n/k \rfloor$ (or $\lceil n/k \rceil$) such that each agent $i$ in $S$ has $u_i(S) > \alpha u_i(X) + \beta$.\footnote{This definition is informal. A precise technical definition is provided in Section \ref{sec:prelims}.} Our main result is the following:
\begin{restatable}{theorem}{thmmainresultrandomized}\label{thm:main-result-randomized}
There is an efficient randomized algorithm which computes a balanced partition that is
$O\left(\sqrt{\frac{\Delta}{k}\ln(k\Delta)}+\ln(k\Delta)\right)$-approximately envy-free
and belongs to the $(k+\sqrt{k},O(k^{3/2}\ln(k\Delta)))$-core.
\end{restatable}

Both approximation factors in the above theorem are independent of $n$, and depend exclusively on the max degree $\Delta$ and the number of parts $k$.
We also obtain a deterministic algorithm with a slightly weaker logarithmic dependence.

\begin{restatable}{theorem}{thmmainresultdeterministic}\label{thm:main-result-deterministic}
There is an efficient deterministic algorithm which computes a balanced partition that is $O\left(\sqrt{\frac{\Delta}{k}\ln(n)}+\ln(n)\right)$-\hspace{0pt}approximately envy-free
and belongs to the $(k+\sqrt{k},O(k^{3/2}\ln(n)))$-core.
\end{restatable}

\paragraph{Main Proof Technique}
The key technical ingredient used in proving both theorems is the probabilistic method. 
The main randomized construction groups the nodes arbitrarily into blocks of size $k$, and independently (and uniformly at random) assigns distinct part labels within each block. This creates a balanced partition. For each fixed node $i$ and part $j$, the number of neighbors of $i$ assigned to $j$ is a sum of independent Bernoulli random variables with mean $\deg(i)/k$. Therefore, using a standard tail bound combined with a union bound over all $n$ nodes proves the existence of a partition with fairness guarantees similar to Theorem \ref{thm:main-result-deterministic}, where there is a logarithmic dependence on $n$. We can improve on this by using the fact that the ``bad'' event for each node $i$ only depends on $O(k\Delta^2)$ other ``bad'' events. Therefore, instead of a union bound, we can use the Lovász local lemma to prove the existence of a partition with the guarantees of Theorem \ref{thm:main-result-randomized}. This existence result can be easily converted to a randomized algorithm using the algorithmic version of the Lovász local lemma. 

To prove Theorem \ref{thm:main-result-deterministic}, we instead sample directly from the uniform distribution over balanced partitions and derandomize this construction using conditional expectations.

\paragraph{Comparison with Prior Work}
Usually, when multiple fairness objectives are achieved simultaneously, there is some loss in the approximation factor of each fairness objective as opposed to the best approximation factor achievable when only considering one such fairness notion. However, there are several cases where our approximation factors surpass the previous best approximations for this problem. 

When finding envy-free balanced partitions (without considering the core), the best known approximation factor is $O \left (\sqrt{\frac{n}{k} \ln k} \right )$. In comparison, our approximation factor of $O\Bigl (\sqrt{\frac{\Delta}{k}\ln(k\Delta)}+\ln(k\Delta)\Bigr)$ is significantly better when $\Delta$ (the max degree of the graph) is smaller than $n$. 

Indeed, graphs where $\Delta \ll n$ are natural in the realm of social networks. Recent work \citep{Agarwal2025Harmonious} motivated by this observation improves the envy-freeness approximation for the special case of $k = 2$ to $\max\{\Delta-2,2\}$. Theorem \ref{thm:main-result-randomized} improves upon this result as well since $O\left(\sqrt{\frac{\Delta}{k}\ln(k\Delta)}+\ln(k\Delta)\right)$ is asymptotically better than $\max(\Delta-2,2)$ when $k = 2$. 

When it comes to the core (without considering envy-freeness), the best known approximation guarantee for general $k$ is $(k,k-1)$~\cite{Li2023GraphPartitioning}.\footnote{They also prove a $(2k-1,0)$ approximation for $n \ge k^2+k$ and a $(1,k)$ approximation for $n < k^2+k$.} Our result achieves an approximation of $(k + \sqrt{k}, O(k^{3/2} \ln(\Delta k)))$, which is only slightly worse, but is achieved simultaneously with an improved envy-freeness approximation.

Finally, we briefly highlight the significance of achieving both these guarantees simultaneously. The difficulty in proving such a result is best captured by the following quote from \cite{Li2023GraphPartitioning}:
\begin{quote}
    {\em there are instances in which every partition in the core achieves the worst possible approximation with respect to envy-freeness, and every partition that achieves the best approximations with respect to envy-freeness, achieves the worst possible approximation with respect to the core.}
\end{quote}
Our results show that if we slightly compromise on the best possible approximation factor of one of the two criteria, we can achieve reasonable approximations for {\em both} the envy-freeness and the core. 

\paragraph{The Special Case of $k = 2$} For the case of $k = 2$, \cite{Li2023GraphPartitioning} present two open questions; we make progress on the first and resolve the second.

\begin{openquestion}\label{oq:one}
When $k = 2$, does every graph admit a balanced partition in the $(1, 0)$-core?
\end{openquestion}

\begin{openquestion}\label{oq:two}
When $k = 2$, can a balanced partition in the $(2, 0)$-core be computed in polynomial time?
\end{openquestion}

On the first question, \cite{Li2023GraphPartitioning} show that the min bisection partition is guaranteed to be in the $(2, 0)$-core. We show that there are instances where the min bisection partition is not in the $(1.56, 0)$-core, establishing a lower bound for this approach. We then prove that every graph with at least $240$ nodes admits a balanced partition in the $(38/25,0)$-core. In addition to improving on the best known approximation factor for the problem, this result has the added significance of breaking past the min bisection ``barrier''.

On the second question, we show that a balanced partition in the $(2,0)$-core can be computed in polynomial time. This result resolves Open Question \ref{oq:two}.

\subsection{Related Work}
Our specific model of dividing agents into coalitions with identical size has been studied in the following three papers \citep{Li2023GraphPartitioning,Agarwal2025Harmonious,Deligkas2025Balanced}. \cite{Li2023GraphPartitioning} introduce the problem and present both existential and algorithmic results on the core and envy-freeness. \cite{Agarwal2025Harmonious} improve on the results of \cite{Li2023GraphPartitioning} for specific classes of graphs like bounded-degree graphs and grid graphs. They also introduce and study the notion of Pareto optimality for this problem. \cite{Deligkas2025Balanced} adapt several notions of fairness from the fair division community to this specific problem, and analyze their existence and computational complexity.

There are also several works in the hedonic games community which study the partitioning of agents into groups with some constraints on their size. Examples of this include \cite{MonacoMoscardelli2023NashBoundedFHG} who study Nash stability (similar to envy-freeness), and \cite{CohenAgmon2024OnlineFriendsPartitioning} who study the problem in an online setting; in both papers, there is an upper bound placed on group size. \cite{BullingerDunajskiElkindGilboa2025SingleDeviationStability} study Nash stability and a few variants when there are both lower bounds and upper bounds on group size. \cite{BiloMonacoMoscardelli2022HedonicFixedSize} study hedonic games when each coalition must have a specific size (not necessarily identical) that is given as part of the input; in this model, they study swap stability where a partition must be robust to two agents performing a mutually beneficial swap in group assignment.

An alternate line of work which studies coalition formation of fixed size is the stable roommates problem \citep{Irving1985StableRoommates,McKay2021Stable}. The main focus in this line of work has been finding stable coalitions of small size ($\le 3$). Our problem, particularly the problem of finding the core, can be viewed as a generalization of this problem with symmetric binary preferences to coalitions of larger size.

Our problem is also related to the problem of balanced graph partitioning, where an undirected graph must be divided into $k$ parts of equal size such that the size of the cut is minimized \citep{AndreevRaecke06BalancedPartitioning,KrauthgamerNaorSchwartz09,BansalEtAl11}. However, most of the literature on this problem focuses on minimizing the global cut size rather than ensuring fairness for the nodes. A few papers study the problem of finding {\em satisfactory} partitions, which are partitions that guarantee each node a certain number of neighbors on its side of the cut \citep{BazganTuzaVanderpooten06,Ciccarelli22}. However, the focus here is on deciding if such a partition exists rather than proving their universal existence.

\section{Preliminaries}\label{sec:prelims}
We have an $n$-node undirected graph $G = (V, E)$ with max degree $\Delta$, and an integer $k$ satisfying $2\le k\le n$. We assume that $\Delta\ge 1$; otherwise, every balanced partition is envy-free and belongs to the $(1,0)$-core. 
For a node $i$ and a set $S \subseteq V$, we use $E(i, S)$ to denote the number of edges from $i$ to the set of nodes $S$ in the graph $G$. Given two sets $S, T \subseteq V$, we use $E(S, T)$ to denote $\sum_{i \in S} E(i, T)$. We use $\deg(i)$ to denote the degree of node $i$ in the graph i.e., the value $E(i, V)$.

We use $X = (X_1, \dots, X_k)$ to denote a $k$-partition of the vertex set $V$. We call $X$ balanced if every part has size either $\lfloor n/k\rfloor$ or $\lceil n/k\rceil$.

Given a partition $X$ and a node $i$, we define the {\em utility} of the node $i$, denoted $u_i(X)$, as the number of edges $i$ has to nodes in its own part. That is, $u_i(X) = E(i, X_j)$ where $i \in X_j$. More generally, given a set $S$, we define the utility of $i$ for the set $S$ (denoted $u_i(S)$) as the value $E(i, S)$.

\subsection{Core}
Given a $k$-partition $X$, a set of nodes $S \subseteq V$ is said to be an $(\alpha, \beta)$-blocking set if $u_i(S) > \alpha \cdot u_i(X) + \beta$ for all $i \in S$.

A balanced partition is said to be in the $(\alpha, \beta)$-core if it has no $(\alpha, \beta)$-blocking sets of size either $\left \lfloor n/k \right \rfloor$ or $\left \lceil n/k \right \rceil$.
\subsection{Envy-Freeness}

Given a $k$-partition $X$, a node $i$ is said to be $r$-envy-free (or EF-$r$) if for all $j \in [k]$, $u_i(X_j) - u_i(X) \le r$. A $k$-partition $X$ is said to be $r$-envy-free if all nodes $v \in V$ are $r$-envy-free.

We note that our definition of envy-freeness is slightly different from \cite{Li2023GraphPartitioning} who define it based on the utility a node would gain if it {\em swapped} its assigned part with another node. Our definition is stronger in the sense that any EF-$r$ partition according to our definition is trivially EF-$r$ according to the definition of \cite{Li2023GraphPartitioning}. We choose our definition since its simplicity makes proofs easier to understand.

\subsection{Probability Tools}
We require the following standard results from the probability literature. The first result is Bernstein's Inequality. We require two separate but very similar versions which follow from \cite[Theorem 2.10]{BoucheronLugosiMassart2013}.

\begin{theorem}[Bernstein's inequality for Bernoulli Variables~\cite{BoucheronLugosiMassart2013}]\label{thm:bernstein-bernoulli}
Let $X_1,\dots,X_n$ be independent Bernoulli random variables. Define $X = \sum_{i = 1}^n X_i$ and $\mu = \sum_{i = 1}^n \E[X_i]$. The following holds for all $t > 0$:
\begin{align*}
\Pr\left[|X-\mu| > \sqrt{2\mu t}+\frac{t}{3} \right] \leq 2e^{-t}.
\end{align*}
\end{theorem}

\begin{theorem}[Bernstein's inequality for Hypergeometric Variables~\cite{BoucheronLugosiMassart2013}]\label{thm:bernstein-hypergeom}
Consider a bag of $N$ balls, of which $k$ are green and $N-k$ are red. Assume we draw $n$ balls from this bag without replacement. Let $X$ denote the number of green balls drawn. The following holds for all $t > 0$:
\begin{align*}
\Pr\left[ \left |X-\frac{k n}{N} \right | > \sqrt{\frac{2knt}{N}}+\frac{t}{3} \right] \leq 2e^{-t}.
\end{align*}
\end{theorem}

Next we need the algorithmic Lovász local lemma. The version we use is \cite[Theorem 1.2]{MoserTardos2010} restated for the symmetric case.

\begin{theorem}[Algorithmic Lovász Local Lemma~\cite{MoserTardos2010}]\label{thm:lll}
Let $X_1,\allowbreak\ldots,\allowbreak X_n$ be mutually independent random variables, and let
$\mathcal{B}=\{B_1,\ldots,B_m\}$ be a finite collection of bad events,
where each event $B\in\mathcal{B}$ is determined by a subset of variables $\operatorname{vbl}(B)\subseteq \{X_1,\ldots,X_n\}$. 

Define the dependency graph $G$ on vertex set $\mathcal{B}$ by joining
two distinct events $B,B'\in\mathcal{B}$ whenever $\operatorname{vbl}(B)\cap \operatorname{vbl}(B')\neq \emptyset$.
Suppose that, for some $p\in[0,1]$ and integer $d\geq 1$, we have $\max_{B \in \cal{B}} \Pr[B] \le p$ and the dependency graph $G$ has max degree at most $d$. 

If $e p (d+1) \le 1$, there exists an assignment of values to the variables $X_1, \dots, X_n$ such that no event in $\cal B$ occurs. Moreover, this assignment can be found by a randomized algorithm that runs in (expected) time polynomial in $m$, $n$, the time it takes to sample a variable $X_i$ and the time it takes to verify if a bad event $B \in \cal B$ has occurred.
\end{theorem}

\section{Approximately Envy-free Partitions in the Approximate Core}\label{sec:approx-existence}
In this section, we prove Theorems \ref{thm:main-result-randomized} and \ref{thm:main-result-deterministic}. 
Our proofs have the following structure: we identify some desirable properties that if a partition satisfies, then it must also be approximately envy-free and in the approximate core. We then show that there exists some partition that satisfies the required desirable properties. We formalize this notion of desirability below.

\begin{definition}[Desirable Partitions]
A $k$-partition $X$ is $(\gamma, \delta)$-{\em desirable} if the following conditions hold:
\begin{enumerate}[(i)]
    \item $X$ is balanced, and
    \item For all nodes $i$ and all parts $j \in [k]$,  $$u_i(X_j) \in \left [\tfrac{\deg(i)}{k} - \left ( \sqrt{\frac{\gamma \deg(i)}{k}} + \delta \right ), \tfrac{\deg(i)}{k} + \left ( \sqrt{\frac{\gamma \deg(i)}{k}} + \delta \right ) \right ].$$
\end{enumerate}
\end{definition}

Intuitively, in desirable partitions, all high degree nodes are in some sense indifferent between the parts. In all of our results, $\gamma$ and $\delta$ will be $\widetilde{O}(1)$. One may think of the term $\left ( \sqrt{\frac{\gamma \deg(i)}{k}} + \delta \right )$ as a small error term. Using this intuition, we show that any desirable partition satisfies the guarantees we seek to prove. 

\begin{theorem}\label{thm:desirable-properties}
For any graph $G$, $(\gamma, \delta)$-desirable partitions are $2 \left ( \sqrt{\frac{\gamma \Delta}{k}} + \delta \right )$-envy-free and in the $(k + \sqrt{k}, k^{3/2} \gamma + 2k \delta)$-core.
\end{theorem}
\begin{proof}
The envy-freeness guarantee follows from the second condition of $(\gamma, \delta)$-desirability.

To prove that the partition is approximately in the core, we use the fact that the utility of each node is upper bounded by its degree in the graph. Assume for contradiction that there exists a $(k + \sqrt{k}, k^{3/2} \gamma + 2 k \delta)$-blocking set $S$ to a $(\gamma, \delta)$-desirable partition $X$. 

Consider some node $i \in S$. Note that $u_i(S) > (k+ \sqrt{k})u_i(X) + 2k\delta + k^{3/2}\gamma$. Since $i$ satisfies the second condition of $(\gamma, \delta)$-\hspace{0pt}desirability, we have $u_i(X) \ge \frac{\deg(i)}{k} - \left ( \sqrt{\frac{\gamma \deg(i)}{k}} + \delta \right )$. We establish a contradiction using the following sequence of inequalities.
\begin{align*}
    u_i(S) &> (k+ \sqrt{k})u_i(X) + 2k\delta + k^{3/2}\gamma \\
    &\ge (k+ \sqrt{k}) \left [\frac{\deg(i)}{k} - \left ( \sqrt{\frac{\gamma \deg(i)}{k}} + \delta \right ) \right ] + 2k\delta + k^{3/2}\gamma \\
    &\ge \deg(i) + \frac{\deg(i)}{\sqrt{k}} - 2 \sqrt{\gamma \deg(i) k} - 2k \delta  + 2k\delta + k^{3/2}\gamma \\
    &= \deg(i) + \frac{\deg(i)}{\sqrt{k}} - 2 \sqrt{\gamma \deg(i) k} + k^{3/2}\gamma \\
    &\ge \deg(i) \ge u_i(S).
\end{align*}
In the second inequality we used the second condition of $(\gamma, \delta)$-desirable partitions, and in the fourth inequality we used the fact that the quadratic expression (in $\sqrt{\deg(i)}$), $\frac{\deg(i)}{\sqrt{k}} - 2 \sqrt{\gamma \deg(i) k} + k^{3/2}\gamma$ is non-negative. This creates a contradiction, which implies that the blocking set $S$ cannot exist.
\end{proof}

We remark that in the above proof, we did not use the size constraint on the blocking set $S$; we only use the loose upper bound that the utility of any agent is upper-bounded by its degree. 

Next, we find parameters $\gamma$ and $\delta$ for which a $(\gamma, \delta)$-desirable partition exists (and can be computed). We prove the following two theorems.

\begin{theorem}\label{thm:desirable-existence-randomized}
A $(12\ln(k\Delta), 2\ln(k\Delta))$-desirable partition always exists, and can be computed using an efficient randomized algorithm.
\end{theorem}

\begin{theorem}\label{thm:desirable-existence-deterministic}
A $(8\ln(n), 4\ln(n))$-desirable partition always exists, and can be computed using an efficient deterministic algorithm.
\end{theorem}

Combining Theorem \ref{thm:desirable-existence-randomized} and Theorem \ref{thm:desirable-properties}, we obtain Theorem \ref{thm:main-result-randomized}. Similarly, combining Theorem \ref{thm:desirable-existence-deterministic} and Theorem \ref{thm:desirable-properties}, we obtain Theorem \ref{thm:main-result-deterministic}.

The remainder of this section is devoted to proving Theorems \ref{thm:desirable-existence-randomized} and \ref{thm:desirable-existence-deterministic}.

\subsection{Proof of Theorem \ref{thm:desirable-existence-randomized}}\label{sec:proof-randomized}
As mentioned in the introduction, this proof uses the probabilistic method. We construct a random partition and show that it is $(12\ln(k\Delta), 2\ln(k\Delta))$-desirable with non-zero probability. This random partition is constructed as follows: partition $V$ arbitrarily into blocks of size $k$, except possibly for one residual block of size less than $k$. Then, independently for each block $B$, choose a uniformly random permutation $\pi_B:B\to [k]$, and assign each node $i \in B$ to part $\pi_B(i)$; for the residual block $B$, let $\pi_B$ be a uniformly random injective function that maps $B$ to $[k]$. Every full block contributes exactly one node to each part, while the residual block contributes at most one node to each part. Thus every possible outcome is balanced, and satisfies condition (i) of $(12\ln(k\Delta), 2\ln(k\Delta))$-desirability.

To show that the probability of satisfying condition (ii) of $(12\ln(k\Delta), 2\ln(k\Delta))$-desirability is non-zero, we use Bernstein's inequality. Specifically, let $Y_{i, j, B}$ be a binary variable indicating whether node $i$ has a neighbor in $X_j \cap B$ where $X_j$ is the $j$-th part of $X$ and $B$ is a block. It is easy to see that $\E[Y_{i, j, B}] = \frac{E(i, B)}{k}$, and $u_i(X_j) = \sum_{B} Y_{i, j, B}$. Therefore, we can apply Theorem \ref{thm:bernstein-bernoulli} to prove the following lemma.

\begin{restatable}{lemma}{lembernsteinapplication}\label{lem:bernstein-application}
For every node $i$ and part $j$, 
\begin{align*}
    \Pr \left [ \left |u_i(X_j) - \tfrac{\deg(i)}{k} \right | > \sqrt{\frac{12\deg(i)\ln(k \Delta)}{k}} + 2\ln(k\Delta) \right ] \le \frac{1}{2e k^2 \Delta^2}.
\end{align*}
\end{restatable}
\begin{proof}[Proof of Lemma]

We have $\E[Y_{i, j, B}] = \frac{E(i, B)}{k}$ and therefore, 
\begin{align*}
    \E[u_i(X_j)] = \sum_{B} \frac{E(i, B)}{k} = \frac{\deg(i)}{k}.
\end{align*}

We can apply Theorem \ref{thm:bernstein-bernoulli} with $t = \ln(4e k^2 \Delta^2)$ to conclude 
\begin{align*}
    \Pr \left [ \left | u_i(X_j) - \tfrac{\deg(i)}{k} \right | > \sqrt{\frac{2\deg(i)\ln(4e k^2 \Delta^2)}{k}} + \frac{\ln(4e k^2 \Delta^2)}{3} \right ] \le \frac{1}{2ek^2 \Delta^2}.
\end{align*}

To prove the lemma, it suffices to show that 
\begin{align*}
\sqrt{\frac{2\deg(i)\ln(4e k^2 \Delta^2)}{k}} + \frac{\ln(4e k^2 \Delta^2)}{3} \le \sqrt{\frac{12\deg(i)\ln(k \Delta)}{k}} + 2\ln(k\Delta).
\end{align*}

This trivially holds when $k \ge 2$.
\end{proof}

Next, we invoke the Lovász local lemma to show that the second condition of $(12\ln(k\Delta), 2\ln(k\Delta))$-desirability is satisfied for some partition. Let $A_i$ denote the {\em bad} event for node $i$, which occurs if node $i$ does not satisfy the second condition of $(12\ln(k\Delta), 2\ln(k\Delta))$-desirability. Using a union bound with Lemma \ref{lem:bernstein-application} we get for each node $i \in V$,
\begin{align*}
    \Pr[A_i] \le \frac{1}{2e k \Delta^2}.
\end{align*}

Next, we consider the number of other bad events which $A_i$ depends on. $A_i$ only depends on the blocks containing the neighbors of $i$. Therefore, each $A_i$ only depends on at most $\Delta$ blocks. Additionally, a fixed block $B$ affects at most $\sum_{i\in B}\deg(i)\le k\Delta$ bad events. Consequently, $A_i$ depends on at most $k\Delta^2$ other bad events.

We apply the algorithmic Lovász local lemma (Theorem \ref{thm:lll}) with $d = k\Delta^2$ and $p = \frac{1}{2e k \Delta^2}$. We can conclude that there is at least one partition which satisfies both conditions of desirability. Moreover, this partition can be computed efficiently using a randomized algorithm. This completes the proof.

\subsection{Proof of Theorem \ref{thm:desirable-existence-deterministic}}\label{sec:proof-deterministic}

To prove Theorem \ref{thm:desirable-existence-deterministic}, sample directly from the uniform distribution over balanced partitions. 
Fix part sizes $m_1,\ldots,m_k\in\{\lfloor n/k\rfloor,\lceil n/k\rceil\}$ that sum to $n$. A uniformly random labeled balanced partition with these part sizes can be generated by taking a uniformly random permutation of $V$ and using consecutive blocks of lengths $m_1,\ldots,m_k$.
For each node $i$ and part $j$, the value of $u_i(X_j)$ follows a hypergeometric distribution. Therefore, we can use Theorem \ref{thm:bernstein-hypergeom} to prove tail bounds for this value.

\begin{restatable}{lemma}{lembernsteinapplicationtwo}\label{lem:bernstein-application-two}
For every node $i$ and part $j$, 
\begin{align*}
    \Pr \left [ \left |u_i(X_j) - \tfrac{\deg(i)}{k} \right | >  \left ( \sqrt{\frac{8\deg(i)\ln(n)}{k}} + 4\ln(n) \right ) \right ] \le \frac{1}{2 nk}.
\end{align*}
\end{restatable}
\begin{proof}[Proof of Lemma]
This follows from a direct application of Theorem \ref{thm:bernstein-hypergeom}. One may think of the green balls as the nodes adjacent to $i$ and the red balls as all the other nodes. Applying Theorem \ref{thm:bernstein-hypergeom} with $t = \ln(4nk)$ immediately gives us
\begin{align*}
    \Pr \left [ \left | u_i(X_j) - \tfrac{\deg(i)}{k} \right | > \sqrt{\frac{2\deg(i)\ln(4nk)}{k}} + \frac{\ln(4nk)}{3} + 1 \right ] \le \frac{1}{2nk}.
\end{align*}

The last $+1$ term accounts for the fact that part sizes are either $\left \lceil \frac n k \right \rceil$ or $\left \lfloor \frac n k \right \rfloor$, rather than the ideal $\frac n k$.
To prove the lemma, it suffices to show that 
\begin{align*}
\sqrt{\frac{2\deg(i)\ln(4nk)}{k}} + \frac{\ln(4nk)}{3} + 1 \le  \left ( \sqrt{\frac{8\deg(i)\ln(n)}{k}} + 4\ln(n) \right ).
\end{align*}

This trivially holds when $n \ge 2$ and $n \ge k$.
\end{proof} 

Using a union bound, we can conclude that the probability that the second condition of $(8\ln(n), 4\ln(n))$-desirability holds is at least $\frac12$. 

This shows that a $(8\ln(n), 4\ln(n))$-desirable partition is guaranteed to exist. It also implies a trivial randomized algorithm to compute such a partition. 
This randomized algorithm can be derandomized using the method of conditional expectations. This is a straightforward argument that we sketch below.

Let $Y$ be the number of node-part pairs $(i,j)$ that violate the second condition of $(8\ln(n), 4\ln(n))$-desirability in our random sampling procedure. We have shown above that $\E[Y] \le \frac12$. 

We start with an empty partition, where no node is assigned to any part. Then we assign nodes one by one such that $\E[Y]$ conditioned on our assignment does not increase. We can do this because we can compute the probability that any pair $(i, j)$ violates the second condition of $(8\ln(n), 4\ln(n))$-desirability in polynomial time since it is the sum of ratios of binomial coefficients. Therefore, every node has a part assignment that does not increase $\E[Y]$. When we terminate, $Y$ is a nonnegative integer satisfying $Y\le 1/2$; this implies $Y=0$.

\section{Approximate Core when \texorpdfstring{$k = 2$}{k = 2}}\label{sec:two-k}
In this section, we make progress on Open Question 1 and resolve Open Question 2 of \cite{Li2023GraphPartitioning}.

The key technique used to prove approximations to the core in prior work has been the min bisection partition \cite{Li2023GraphPartitioning}. The min bisection partition $(X_1, X_2)$ is one that minimizes the cut size subject to being a balanced partition. \cite{Li2023GraphPartitioning} show that the min bisection partition is in the $(2, 0)$-core. 

We present an improved result showing that the min bisection partition is in the $\left (\frac{\sqrt{17}-1}{2} + o(1), 0 \right )$-core. 

\begin{theorem}\label{thm:min-bisection-core}
The min bisection partition $X$ is in the $\bigl(\frac{\sqrt{17}-1}{2} + \frac{100}{n}, 0\bigr)$-core.
\end{theorem}
\begin{proof}
Throughout this proof, we use $\alpha$ to denote the value $(\sqrt{17}-1)/2$. We also assume $n \ge 50$; otherwise $\frac{\sqrt{17}-1}{2} + \frac{100}{n} \ge 2$, and the result of \cite{Li2023GraphPartitioning} proves that the min bisection partition is in the $(2, 0)$-core. 

Assume for contradiction that the min bisection partition has a $\left(\frac{\sqrt{17}-1}{2} + \frac{100}{n}, 0\right)$-blocking set $S$. We define a mathematical program consisting of constraints that the graph must satisfy if $S$ exists. We then show that this mathematical program is infeasible, proving that $S$ cannot exist. Crucially, the variables and constraints of this program will not depend on the number of nodes of the graph $n$. To define this mathematical program, we divide the set of nodes into four blocks depending on its membership in $X_1$ and $S$. Formally, we define the four sets
\begin{align*}
A=X_1\cap S,\quad B=X_1\setminus S,\quad C=X_2\cap S,\quad D=X_2\setminus S.
\end{align*}

We define the notation $a = |A|$ and $c = |C|$. We assume without loss of generality that $a \ge c \ge 1$. Using this notation, we define the constraints of our mathematical program.

Summing over the blocking inequalities for all nodes in $A$, we get
\begin{align*}
E(A,A)+E(A,C)&> \left (\alpha + \frac{100}{n} \right ) E(A,X_1).
\end{align*}

Since $E(A,X_1)=E(A,A)+E(A,B)$, this simplifies to
\begin{align}
E(A,C)&> \left (\alpha + \frac{100}{n} \right )  E(A,B) + \left (\alpha - 1 + \frac{100}{n} \right ) E(A, A) \notag\\
&> \left (1 + \frac{50}{n} \right ) \left [\alpha E(A, B) + (\alpha - 1) E(A, A) \right ]. \label{eq:a-blocking-inequality}
\end{align}

We can prove a similar inequality for the set $C$.
\begin{align}
E(A,C) > \left (1 + \frac{50}{n} \right ) \left [\alpha E(C, D) + (\alpha - 1) E(C, C) \right ]. \label{eq:c-blocking-inequality}
\end{align}

Next, we exploit the fact that $X$ is a min bisection partition. To begin with, the partition $(S, V \setminus S)$ cannot have a smaller cut size. Therefore, 
\begin{align}
E(A,B)+E(C,D)\ge E(A,C)+E(B,D)\ge E(A,C). \label{eq:s-larger-cut}
\end{align}

We establish two additional constraints that follow from swapping some nodes in $A$ with an equal number of nodes in $C$. Swapping any $i \in A$ with any $j \in C$ cannot decrease the cut size. Therefore,
\begin{align*}
E(i,X_1)+E(j,X_2)\ge E(i,C)+E(j,A)-2E(i,j).
\end{align*}

Summing up this inequality over all pairs $(i, j) \in A \times C$, we get
\begin{align}
c\left (E(A,A) + E(A, B) \right ) + a \left (E(C,C) + E(C, D) \right )\ge (|S|-2) E(A,C). \label{eq:local-swap}
\end{align}

If we swap $C$ with a uniformly random $c$-subset of $A$, the cut size must increase. This gives us the following inequality 
\begin{align}
\frac{c(a-c)}{a(a-1)}E(A,A) +\frac ca E(A,B)+E(C,D) \ge \left (1-\frac ca \right )E(A,C). \label{eq:larger-swap}
\end{align}

In summary, the mathematical program consists of inequalities \eqref{eq:a-blocking-inequality}, \eqref{eq:c-blocking-inequality}, \eqref{eq:s-larger-cut}, \eqref{eq:local-swap} and \eqref{eq:larger-swap} along with the basic constraints $a \ge c \ge 1$ and $a + c = |S| \in \left \{ \left \lfloor \frac n2 \right \rfloor, \left \lceil \frac n2 \right \rceil \right \}$.

In the remainder of this proof, we show that the mathematical program defined above is infeasible. This shows that the set $S$ cannot exist.
Since $|S|\ge25$ and $n\le 2|S|+1$, we have
\begin{align*}
\left(1+\frac{50}{n}\right)^{-1}<1-\frac{3}{|S|}.
\end{align*}

The two blocking constraints therefore imply
\begin{align*}
\alpha E(A,B)+(\alpha-1)E(A,A) &<\left(1-\frac{3}{|S|}\right)E(A,C),\\
\alpha E(C,D)+(\alpha-1)E(C,C) &<\left(1-\frac{3}{|S|}\right)E(A,C).
\end{align*}

Multiplying \eqref{eq:local-swap} by $\alpha-1$ and using these bounds gives
\begin{equation}\label{eq:core-weighted-edges}
\begin{aligned}
cE(A,B)+aE(C,D)
&<\bigl((2-\alpha)|S|-3+2(\alpha-1)\bigr)E(A,C)\\
&<(2-\alpha)|S|E(A,C).
\end{aligned}
\end{equation}
By \eqref{eq:s-larger-cut}, the left-hand side is at least $aE(A,C)-(a-c)E(A,B)$. Hence
\begin{equation}\label{eq:core-ab-lower}
(a-c)E(A,B)
>\bigl(a-(2-\alpha)|S|\bigr)E(A,C).
\end{equation}

Applying the first blocking bound to \eqref{eq:core-ab-lower} yields
\begin{equation}\label{eq:core-aa-upper}
\begin{aligned}
(\alpha-1)(a-c)E(A,A)
&<\left(1-\frac{3}{|S|}\right)(a-c)E(A,C)-\alpha(a-c)E(A,B)\\
&<\frac{2-\alpha}{2}\left (a-c-(2-\alpha)|S| \right )E(A,C).
\end{aligned}
\end{equation}
The last step uses $a+c=|S|$ and $\alpha^2+\alpha=4$. The left-hand side is nonnegative, so \eqref{eq:core-aa-upper} forces $a-c-(2-\alpha)|S| >0$. Specifically, this implies that $a > c$.

On the other hand, multiplying \eqref{eq:larger-swap} by $a$ and using \eqref{eq:core-weighted-edges} gives
\begin{equation}\label{eq:core-aa-lower}
(a-c-(2-\alpha)|S|)E(A,C)<\frac{c(a-c)}{a-1}E(A,A).
\end{equation}
Combining \eqref{eq:core-aa-upper} and \eqref{eq:core-aa-lower} now requires
\[
\frac{(\alpha-1)(a-1)}{c}<\frac{2-\alpha}{2}.
\]
This is impossible since $a -1 \ge c$.
\end{proof}

We further show that the above bound is tight. That is, there are instances where the min bisection partition is not in the $\left (\frac{\sqrt{17}-1}{2} - \vare, 0 \right )$-core for any constant $\vare > 0$. 
The complete proof can be found in Appendix \ref{apdx:two-k}.

\begin{restatable}{theorem}{thmminbisectionlimit}\label{thm:min-bisection-limit}
For every $\vare>0$, there are graphs where the unique min bisection partition is not in the $\left(\frac{\sqrt{17}-1}{2}-\vare,0\right)$-core.
\end{restatable}
\begin{proof}[Proof sketch]
We construct a random graph over $2N$ nodes (for some sufficiently large $N$). We assume without loss of generality that $\vare < 0.01$. We partition the nodes into four sets $A,B,C,D$ such that $|A|=|C|\approx \left (\frac{\sqrt{17}-3}{4} \right )N$ and $|B|=|D| \approx \left ( \frac{7-\sqrt{17}}{4} \right ) N$.

The edges are constructed as follows. We make $A$ and $D$ cliques. The remaining edges are added with carefully chosen probabilities. Let $\delta$ be a sufficiently small constant. We independently include each possible edge between $A$ and $B$ with probability $\frac{\sqrt{17}+1}{8}\delta$, between $B$ and $C$ with probability $\delta(1-\vare)$, within $C$ with probability $2\delta$, and between $C$ and $D$ with probability $\frac{7-\sqrt{17}}{8}\delta$. No other edges are included. With this random construction, we can show that $X=(A\cup B,C\cup D)$ is the unique min bisection partition with high probability, and that $B\cup C$ is a $\left(\frac{\sqrt{17}-1}{2}-3\vare,0\right)$-blocking set to $X$ with high probability. Rescaling $\vare$ proves the result.
\end{proof}

The value of $\frac{\sqrt{17}-1}{2}$ is roughly $1.5615$. Our next result proves that we can go past this barrier if we consider a slight refinement to this approach. 

\begin{restatable}{theorem}{thmtwocoreimproved}\label{thm:two-core-improved}
When $k=2$, every graph with at least $240$ nodes admits a balanced partition in the $(38/25,0)$-core.
\end{restatable}
\begin{proof}[Proof Sketch]
We show that if the min bisection partition has a $\left (38/25, 0 \right )$-blocking set $S$, then the partition $(S, V \setminus S)$ is in the $\left (38/25, 0 \right )$-core. 

The proof uses the same high-level approach as Theorem \ref{thm:min-bisection-core}. We assume for contradiction that the partition $(S, V \setminus S)$ has a $\left (38/25, 0 \right )$-blocking set $T$, and then prove that $T$ cannot exist. To do this, we write out a mathematical program consisting of a list of constraints defined by the min bisection partition, its blocking set $S$ and the potential blocking set $T$ to the partition $(S, V \setminus S)$. If this program is infeasible, then $T$ cannot exist. In Theorem \ref{thm:min-bisection-core}, we could analytically prove that the program is infeasible. For this theorem, we require a computer program (in Python) to prove infeasibility. This is mainly due to the size of the program and the complexity of the constraints. All the details can be found in Appendix \ref{apdx:two-core-improved}.
\end{proof}

Finally, we turn to the question of polynomial time computability. The previous result does not immediately imply a polynomial time algorithm to find such a partition. Indeed, finding the initial min bisection $(X_1, X_2)$ is computationally intractable~\cite{garey1976simplified}. In our next result, we show that a simple local search procedure results in a balanced partition in the $(2,0)$-core.

\begin{theorem}\label{thm:two-k-poly-time}
When $k = 2$, every graph admits a balanced partition in the $(2,0)$-core, and this partition can be computed in polynomial time.
\end{theorem}
\begin{proof}
Our algorithm is simple. We first check if the graph has a connected component $C$ of size $\lfloor n/2\rfloor$ or $\lceil n/2\rceil$. If so, we return $(C,V\setminus C)$. This partition is trivially in the $(1, 0)$-core.

Otherwise, let $X=(X_1, X_2)$ be a balanced partition such that swapping a node in $X_1$ with a node in $X_2$ does not reduce the size of the cut $E(X_1, X_2)$. Such a balanced partition can be efficiently computed by starting with any balanced partition and repeatedly performing a swap that strictly decreases the cut. This is exactly equivalent to the local min-cut procedure used by \cite{Li2023GraphPartitioning}. 

Assume for contradiction that there is a $(2,0)$-blocking set $S$ of size $\lfloor n/2\rfloor$ or $\lceil n/2\rceil$. It must intersect both parts, since a subset of one part cannot block $X$. We will show that $S$ is a connected component of size $\lfloor n/2 \rfloor$ or $\lceil n/2 \rceil$. This contradicts the first step of the algorithm which would have output the partition $(S, V \setminus S)$. 

Fix a node $i \in X_1 \cap S$ and $j \in X_2 \cap S$. Swapping $i$ and $j$ cannot decrease the cut size. Therefore, 
\begin{align}
E(i, X_1) + E(j, X_2) \ge E(i, X_2) + E(j, X_1) - 2E(i, j). \label{eq:local-swap-1}
\end{align}

By the definition of a blocking set, $E(i, X_1 \cap S) + E(i, X_2 \cap S) > 2E(i, X_1)$. We add $E(i, X_2 \setminus S)$ to both sides and move some terms to get
\begin{align}
E(i, X_2) -  E(i, X_1)  \ge E(i, X_1 \setminus S)  + E(i, X_2 \setminus S) + 1. \label{eq:local-swap-2}
\end{align}
Here, we switch the strict inequality with a weak one and add a $+1$ term on the right since all the terms are integers. Similarly for $j$ we have
\begin{align}
E(j, X_1) - E(j, X_2) \ge E(j, X_2 \setminus S) + E(j, X_1 \setminus S) + 1. \label{eq:local-swap-3}
\end{align}

Using Equations \eqref{eq:local-swap-2} and \eqref{eq:local-swap-3}, we can simplify \eqref{eq:local-swap-1} to the following

\begin{align*}
2 ( E(i, j) - 1) \ge E(i, X_1 \setminus S) + E(i, X_2 \setminus S) + E(j, X_1 \setminus S) + E(j, X_2 \setminus S).
\end{align*}

In this inequality, the left hand side is upper bounded by $0$ and the right hand side is lower bounded by $0$. Therefore, we must have $E(i, j) = 1$ and all the terms on the right hand side must be equal to $0$. This implies that $E(i, V \setminus S) = E(j, V \setminus S) = 0$. Since we picked $i$ and $j$ arbitrarily from $X_1 \cap S$ and $X_2 \cap S$ respectively, we can conclude that $E(S, V \setminus S) = 0$. Moreover, each node in $X_1 \cap S$ has an edge to each node in $X_2 \cap S$. Therefore, $S$ is a connected component of size either $\lfloor n/2 \rfloor$ or $\lceil n/2 \rceil$. This completes the proof.
\end{proof}

\section{Conclusion}
We study the problem of finding balanced graph partitions that simultaneously satisfy approximate envy-freeness and approximate core stability. We present a randomized algorithm that runs in expected polynomial time and achieves approximation guarantees depending only on the maximum degree $\Delta$ and the number of parts $k$. We also give a deterministic polynomial-time algorithm with a slightly weaker logarithmic dependence on the number of nodes $n$. When $k=2$, we show that a balanced partition in the $(38/25,0)$-core exists for graphs with at least $240$ nodes, and a balanced partition in the $(2,0)$-core can be computed in polynomial time.

The main open question we leave is whether our bounds can be improved. There is no reason to believe that $\widetilde{O}\left (\sqrt{\frac{\Delta}{k}} \right )$ is the best possible approximation factor for envy-freeness, and similarly for the core. Improving these bounds is likely to require novel techniques which may be of independent interest.

\section*{Disclosure of AI Usage}
An earlier version of the manuscript contained human-derived results that proved that the min bisection partition was in the $(1.618, 0)$-core and that a $(2 + \vare, 0)$-core can be computed in polynomial time. ChatGPT Astra (on Ultra/Max thinking) improved the constants in both of these results to $(1.56, 0)$ (Theorem \ref{thm:min-bisection-core}) for the former result and $(2, 0)$ (Theorem \ref{thm:two-k-poly-time}) for the latter result. The key idea for the improvement to $(1.56, 0)$ in Theorem \ref{thm:min-bisection-core} came from the authors and the main contribution of ChatGPT Astra was an infeasibility proof for the mathematical program. 

It further improved the core existence result to $(1.52, 0)$ (Theorem \ref{thm:two-core-improved}) and proved a lower bound on the core guarantees of the min bisection partition (Theorem \ref{thm:min-bisection-limit}). Theorem \ref{thm:two-core-improved} was the result of a long conversation which involved trying different constructions of the mathematical program until we reached one that was satisfactory. The idea of using a mathematical program and the structure of the mathematical program were decided by the authors (over multiple iterations). The main contribution of ChatGPT Astra was to prove that the program was infeasible.

All other results were human-derived. ChatGPT Astra was also used to proofread and polish the writing of the paper. All results were verified by the authors and we take full responsibility for their correctness.

\section*{Acknowledgments}
The authors would like to thank Rik Sengupta for useful discussions.
Vignesh Viswanathan is funded by the National Science Foundation (NSF) Career Award 2441296 and Grant RI-2327057.
Nisarg Shah is supported by an NSERC Discovery Grant and an NSERC-CSE Research Communities Grant. Researchers funded through the NSERC-CSE Research Communities Grants do not
represent the Communications Security Establishment Canada or the Government of Canada. Any research, opinions or positions they produce as part of this initiative do not represent the official views of the Government of Canada.

\bibliographystyle{alpha}
\bibliography{abb,references}

\newpage
\appendix

\section{Missing Proofs from Section \ref{sec:two-k}}\label{apdx:two-k}

\thmminbisectionlimit*
\begin{proof}
We construct a random graph on $2N$ nodes, where $N$ is a sufficiently large integer. We assume without loss of generality that $\vare < 0.01$. We show that this graph has a unique min bisection partition that admits a $\left(\frac{\sqrt{17}-1}{2}-3\vare,0\right)$-blocking set with probability tending to one. Since $\vare$ can be arbitrarily small, this implies the stated result by rescaling $\vare$.

Partition the nodes into four sets $A,B,C,D$ with $|A|=|C|=\left\lceil\frac{\sqrt{17}-3}{4}N\right\rceil$ and $|B|=|D|=N-\left\lceil\frac{\sqrt{17}-3}{4}N\right\rceil$.
For convenience, write $r=|A|/N=|C|/N$ and $p=|B|/N=|D|/N=1-r$. For all sufficiently large $N$, we have $1/4\le r\le1/3$. Fix $0<\delta<(\sqrt{17}-3)/128$. Independently include each unordered pair of distinct nodes as an edge with the probability specified below.

\begin{center}
\renewcommand{\arraystretch}{1.4}
\begin{tabular}{lc}
\toprule
Edge type & Probability \\
\midrule
Within $A$ or within $D$ & $1$ \\
Between $A$ and $B$ & $\frac{\sqrt{17}+1}{8}\delta$ \\
Between $B$ and $C$ & $\delta(1-\vare)$ \\
Within $C$ & $2\delta$ \\
Between $C$ and $D$ & $\frac{7-\sqrt{17}}{8}\delta$ \\
All other pairs & $0$ \\
\bottomrule
\end{tabular}
\end{center}
Thus, $A$ and $D$ are cliques. Write $X= (X_1, X_2) = (A\cup B,C\cup D)$ and $S=B\cup C$; both parts of $X$ and the set $S$ have size $N$. Figure \ref{fig:min-bisection-construction} illustrates the construction, the partition $X$, and the blocking set $S$. There are four pieces of notation defined here that will be used throughout. The sizes of $A$ and $B$ are $rN$ and $pN$ respectively. The constants $\vare$ and $\delta$ can be thought of as very small values. 

\begin{figure}[tb]
\centering
\begin{tikzpicture}[
    font=\small,
    group/.style={draw=blue!35!black, line width=0.65pt, rounded corners=22pt},
    set/.style={circle, draw=black!70, line width=0.7pt, inner sep=0pt, align=center},
    density/.style={font=\scriptsize, fill=white, inner sep=2pt},
    link/.style={draw=black!70, line width=0.8pt}
]
\pgfmathsetmacro{\largediameter}{1.10*sqrt((7-sqrt(17))/(sqrt(17)-3))}
\fill[blue!4, rounded corners=22pt] (-1.18,-1.08) rectangle (1.18,3.76);
\fill[blue!4, rounded corners=22pt] (2.52,-1.08) rectangle (4.88,3.76);
\fill[orange!12, rounded corners=22pt] (-1.30,-1.18) rectangle (5.00,1.08);
\draw[group] (-1.18,-1.08) rectangle (1.18,3.76);
\draw[group] (2.52,-1.08) rectangle (4.88,3.76);
\draw[orange!55!black, dashed, line width=0.85pt, rounded corners=22pt]
    (-1.30,-1.18) rectangle (5.00,1.08);
\node[anchor=south] at (0,3.86) {$X_1=A\cup B$};
\node[anchor=south] at (3.70,3.86) {$X_2=C\cup D$};
\node[anchor=north, text=orange!45!black] at (1.85,-1.30) {$S=B\cup C$};

\node[set, fill=black!9, minimum size=1.10cm] (A) at (0,2.65)
    {$A$\\[-1pt]{\scriptsize $rN$}};
\node[set, fill=orange!22, minimum size=\largediameter cm] (B) at (0,0)
    {$B$\\[-1pt]{\scriptsize $pN$}};
\node[set, fill=orange!22, minimum size=1.10cm] (C) at (3.70,0)
    {$C$\\[-1pt]{\scriptsize $rN$}};
\node[set, fill=black!9, minimum size=\largediameter cm] (D) at (3.70,2.65)
    {$D$\\[-1pt]{\scriptsize $pN$}};

\draw[link] (A) -- (B);
\node[density, anchor=west] at (0.12,1.46) {$\frac{\sqrt{17}+1}{8}\delta$};
\draw[link] (B) -- node[density, above=2pt] {$\delta(1-\vare)$} (C);
\draw[link] (D) -- (C);
\node[density, anchor=west] at (3.82,1.46) {$\frac{7-\sqrt{17}}{8}\delta$};
\draw[link] (C.35) .. controls +(0.68,0.50) and +(0.68,-0.50) .. (C.-35);
\node[density] at (4.53,0.70) {$2\delta$};
\end{tikzpicture}
\caption{The construction with bisection $X$ (solid outlines) and blocking set $S$ (dashed outline). Inter-set edges have the indicated probabilities. The loop on $C$ represents internal edges which exist with probability $2\delta$; $A$ and $D$ are cliques while $B$ is independent.}
\label{fig:min-bisection-construction}
\end{figure}

We first lower bound the expected difference between any competing bisection and $X$. This is done in the following lemma that shows that any balanced partition is worse than $X$ and the amount it is worse is proportional to the amount it differs from $X$.

\begin{lemma}\label{lem:min-bisection-expected-gap}
For the random graph above and all sufficiently large $N$, let $Y=(Y_1,Y_2)$ be a fixed balanced partition. Let $X = (X_1, X_2) = (A \cup B, C \cup D)$. Assume without loss of generality that $\tau=|Y_1\mathbin{\triangle}X_1|\le N$. Then
\begin{align*}
    \E[E(Y_1,Y_2)-E(X_1,X_2)] \ge \frac{\delta\vare}{12}\tau N.
\end{align*}
\end{lemma}
\begin{proof}[Proof of the lemma]
Let $a,b,c,d$ denote the numbers of nodes moved from $A,B,C,D$ to the other part of $X$, divided by $N$. That is,
\[
    a=\frac{|A\cap Y_2|}{N},\quad b=\frac{|B\cap Y_2|}{N},\quad
    c=\frac{|C\cap Y_1|}{N},\quad d=\frac{|D\cap Y_1|}{N}.
\]
Balance gives $a+b=c+d=\tau/(2N)\le1/2$. Also, $0\le a,c\le r$ and $0\le b,d\le p$.

Our goal is to lower bound $\E[E(Y_1,Y_2)-E(X_1,X_2)]$. We will instead lower bound the normalized expected difference in cut sizes $\frac{\E[E(Y_1,Y_2)-E(X_1,X_2)]}{N^2}$. To measure this value, we can look at the contribution from each group of edges to this expression. This is given in the following table.

\begin{center}
\renewcommand{\arraystretch}{1.5}
\begin{tabular}{@{}lc@{}}
\toprule
Edge type & Expected cut difference divided by $N^2$ \\
\midrule
Within $A$ & $a(r-a)$ \\
Within $D$ & $d(p-d)$ \\
$A$--$B$ & $\frac{\sqrt{17}+1}{8}\delta\big(a(p-b)+b(r-a)\big)$ \\
$B$--$C$ & $-\delta(1-\vare)\big(b(r-c)+c(p-b)\big)$ \\
Within $C$ & $2\delta c(r-c)$ \\
$C$--$D$ & $\frac{7-\sqrt{17}}{8}\delta\big(c(p-d)+d(r-c)\big)$ \\
\bottomrule
\end{tabular}
\end{center}
Summing the entries in the table, the normalized expected cut difference is
\begin{align*}
    F&=\frac{\E[E(Y_1,Y_2)-E(X_1,X_2)]}{N^2}\\
    &=a(r-a)+d(p-d)+\frac{\sqrt{17}+1}{8}\delta\big(a(p-b)+b(r-a)\big)-\delta(1-\vare)\big(b(r-c)+c(p-b)\big)\\
    &\quad+2\delta c(r-c)+\frac{7-\sqrt{17}}{8}\delta\big(c(p-d)+d(r-c)\big).
\end{align*}
The $B$--$C$ term is negative because these edges are cut by $X$ and become uncut when exactly one endpoint moves. All other terms count edges that become cut when exactly one endpoint moves.

We now simplify this expression using the constraints. Substituting $b=c+d-a$ and $p+r=1$, and collecting the terms proportional to $\vare$ separately gives us
\begin{align*}
    F={}&a(r-a)+d(p-d)+\delta\vare\big(b(r-c)+c(p-b)\big)+2\delta a(r-c)\\
    &+\delta\left(r-\frac{\sqrt{17}+1}{8}(p-r)\right)(b-d)+\frac{\sqrt{17}+1}{4}\delta(d-a)(c-a).
\end{align*}
We bound the last two terms using the ranges of the variables. Since $p \ge r$, we have $r-\frac{\sqrt{17}+1}{8}(p-r)\le r$. Since $b\ge0$, the penultimate term is therefore at least $-\delta rd$. Also, $c-a\ge-r$ and $c-a\le r-a$, so
\[
    (d-a)(c-a)=d(c-a)-a(c-a)\ge-rd-a(r-a).
\]
Using $(\sqrt{17}+1)/4<2$ and dropping the nonnegative term $2\delta a(r-c)$, we obtain
\begin{align*}
    F\ge{}&(1-2\delta)a(r-a)+d(p-d)-3\delta rd+\delta\vare\big(b(r-c)+c(p-b)\big)\\
    \ge{}&(1-2\delta)a(r-a)+d(p-d)-\delta d+\delta\vare\big(b(r-c)+c(p-b)\big),
\end{align*}
where the second inequality uses $r\le1/3$.

Finally, $d\le\tau/(2N)\le1/2$ and $p\ge2/3$ imply $d(p-d)\ge d/6$. This simplifies the expression to
\begin{align*}
    F\ge{}&(1-2\delta)a(r-a)+(1-6\delta)d(p-d)+\delta\vare\big(b(r-c)+c(p-b)\big)\\
    \ge{}&\delta\vare\big(a(r-a)+d(p-d)+b(r-c)+c(p-b)\big).
\end{align*}
The expression within parentheses simplifies further using $a+b=c+d=\tau/(2N)$ and $p+r=1$:
\begin{align*}
    a(r-a)+d(p-d)+b(r-c)+c(p-b)&=\frac{\tau}{2N}-a^2-d^2-2bc\\
    &=\frac{\tau}{2N}-\frac{\tau^2}{4N^2}-(a-c)^2.
\end{align*}
Since $a,c\in[0,r]$ and $a,c\le\tau/(2N)$, we have $(a-c)^2\le r\tau/(2N)$. Together with $r\le1/3$ and $\tau/(2N)\le1/2$, this gives
\begin{align*}
    \frac{\tau}{2N}-\frac{\tau^2}{4N^2}-(a-c)^2
    &\ge\frac{\tau}{2N}\left(1-\frac{\tau}{2N}-r\right)\\
    &\ge\frac{\tau}{12N}.
\end{align*}
Consequently,
\[
    \E[E(Y_1,Y_2)-E(X_1,X_2)]
    \ge\frac{\delta\vare}{12}\tau N,
\]
as required.
\end{proof}

Fix a balanced partition $Y = (Y_1, Y_2)$ with $\tau$ defined as in Lemma \ref{lem:min-bisection-expected-gap}. The random difference $E(Y_1,Y_2)-E(X_1,X_2)$ is a deterministic constant plus a sum of independent Bernoulli random variables, one for each possible edge in the potential cut. Only pairs with exactly one endpoint in $Y_1\mathbin{\triangle}X_1$ can contribute, so there are at most $\tau(2N-\tau)\le2\tau N$ such random variables. As a consequence, the expectation of this sum is at most $2\tau N$.

We apply Theorem \ref{thm:bernstein-bernoulli} with $t=\delta^2\vare^2\tau N/48^2$. Its deviation threshold is at most
\begin{align*}
    \frac{\delta\vare\tau N}{24}
    +\frac{\delta^2\vare^2\tau N}{3\cdot48^2}
    <\frac{\delta\vare\tau N}{12}.
\end{align*}
Lemma \ref{lem:min-bisection-expected-gap} gives an expected cut difference of at least $\delta\vare\tau N/12$. Thus, for the cut difference to be nonpositive, the Bernoulli sum must deviate from its expectation by more than the threshold above. Consequently,
\begin{align*}
    \Pr[E(Y_1,Y_2)\le E(X_1,X_2)]
    \le2\exp\left(-\frac{\delta^2\vare^2}{48^2}\tau N\right).
\end{align*}

Since both parts have size $N$, the distance $\tau$ is even, and there are exactly $\binom{N}{\tau/2}^2\le(2eN/\tau)^\tau$ choices for $Y_1$ at distance $\tau$ from $X_1$. A union bound over all competing bisections gives a failure probability at most
\begin{align*}
    2\sum_{\substack{2\le\tau\le N\\\tau\text{ even}}}
    \exp\left(\tau\ln\frac{2eN}{\tau}
    -\frac{\delta^2\vare^2}{48^2}\tau N\right)=o(1).
\end{align*}

This shows that $X$ is the unique min bisection partition, with high probability tending to one as $N$ grows.

Finally, we show that $S=B\cup C$ is a $\left(\frac{\sqrt{17}-1}{2}-3\vare,0\right)$-blocking set with high probability. For each node $i\in B$,
\begin{align*}
    \E[u_i(X)]=\frac{\sqrt{17}+1}{8}\delta rN,\qquad
    \E[u_i(S)]=\delta(1-\vare)rN.
\end{align*}
The ratio of these expectations is $\frac{\sqrt{17}-1}{2}(1-\vare)>\frac{\sqrt{17}-1}{2}-2\vare$. For each node $i\in C$,
\begin{align*}
    \E[u_i(X)]&=2\delta(rN-1)+\frac{7-\sqrt{17}}{8}\delta pN,\\
    \E[u_i(S)]&=2\delta(rN-1)+\delta(1-\vare)pN.
\end{align*}
For sufficiently large $N$, the ratio of these expectations converges to
\begin{align*}
    \frac{2+(1-\vare)\frac{\sqrt{17}+1}{2}}
    {2+\frac{7-\sqrt{17}}{8}\frac{\sqrt{17}+1}{2}}
    &=\frac{\sqrt{17}-1}{2}-(5-\sqrt{17})\vare\\
    &>\frac{\sqrt{17}-1}{2}-2\vare.
\end{align*}

In summary, both ratios are at least $\frac{\sqrt{17}-1}{2}-2\vare$. Moreover, using $r\ge1/4$, we can conclude that $\E[u_i(X)]\ge\delta N/8$ for every $i\in S$. Thus, we obtain
\begin{align*}
    &\E\left[u_i(S)-\left(\frac{\sqrt{17}-1}{2}-3\vare\right)u_i(X)\right]
    \ge\vare\E[u_i(X)]\ge\frac{\delta\vare N}{8}.
\end{align*}
For each node $i$, both $u_i(X)$ and $u_i(S)$ are sums of independent Bernoulli variables. Therefore, we can apply Theorem \ref{thm:bernstein-bernoulli} to each utility with $t=\delta^2\vare^2N/8192$ to get a deviation threshold of at most
\begin{align*}
    \frac{\delta\vare N}{64}+\frac{\delta^2\vare^2N}{24576}<\frac{\delta\vare N}{32}.
\end{align*}
Taking a union bound over these two tail bounds gives us
\begin{align*}
    \Pr\Big[\max\big\{|u_i(X)-\E[u_i(X)]|,|u_i(S)-\E[u_i(S)]|\big\}>\delta\vare N/32\Big]
    \le4e^{-\delta^2\vare^2N/8192}.
\end{align*}
This proves the following for each $i \in S$,
\[
    \Pr\left[u_i(S)\le\left(\frac{\sqrt{17}-1}{2}-3\vare\right)u_i(X)\right]
    \le4\exp(-\delta^2\vare^2N/8192).
\]
Taking a union bound over the nodes of $S$ shows that $S$ is a $\left (\frac{\sqrt{17}-1}{2} - 3\vare, 0 \right)$-blocking set to $X$ with high probability. Combining this with the first part, there exists some instantiation of our random graph with sufficiently large $N$ where $X$ is the min bisection and $S$ is a $\left (\frac{\sqrt{17}-1}{2} - 3\vare, 0 \right)$-blocking set to $X$. This completes the proof.
\end{proof}

\section{Proof of Theorem~\ref{thm:two-core-improved}}\label{apdx:two-core-improved}

In this section, we prove Theorem~\ref{thm:two-core-improved}. The theorem statement is restated below for convenience.

\thmtwocoreimproved*
Throughout this proof we use $\alpha$ to denote the value $38/25$. Let $X=(X_1,X_2)$ be a min bisection partition. If $X$ is in the $(\alpha, 0)$-core, we are done. Otherwise, let $S$ be a $(\alpha, 0)$-blocking set for $X$. If $(S,V\setminus S)$ is in the $(\alpha,0)$-core, we are done once again. Otherwise, let $T$ be a $(\alpha, 0)$-blocking set for $(S, V \setminus S)$. We show that these two blocking sets cannot both exist.

Note that $S$ must intersect with $T$ as well as $X_1$ and $X_2$. Relabel $X_1,X_2$ so that $X_1$ contains at least as many nodes of $S$ as $X_2$. If the graph has an odd number of nodes, add one isolated dummy node to the smaller side of $X$, and add it to each of $S,T$ that originally had the smaller
allowed size. Every bisection of the augmented graph induces a balanced cut of the original graph with the same cut size. Thus $X$ is still a min bisection, and $X_1,X_2,S,T$ now all have size $N/2$, where $N$ is the number of nodes of the new graph. By the assumption of the theorem statement, $N\ge240$. Define the value $\zeta=1/240$. 

The intersection has an important structure. For a real node $i\in S\cap T$, strict blocking, $\alpha>3/2$, and integrality give
\begin{align*}
2u_i(S)\ge3u_i(X)+1,\qquad 2u_i(T)\ge3u_i(S)+1.
\end{align*}
Thus $4u_i(T)\ge9u_i(X)+5>8u_i(X)+4$, so $\deg(i)\ge u_i(T)\ge2u_i(X)+2$. In particular, moving $i$ alone to the opposite side of $X$ changes the cut by $2u_i(X)-\deg(i)\le-2$. If two real nodes in $S\cap T$ were on opposite sides of the even min bisection, swapping them would change its cut by at most
$-2-2+2<0$. Therefore, in the original graph, a node of $S\cap T$ cannot be on the larger side of $X$, since moving it alone to the smaller side preserves balance and reduces the cut. Hence all real overlap nodes lie on one side of $X$; when present, the added dummy node also lies on that side.

We use a small relaxation of the graph constraints. Divide $X_1$ into
\[
\begin{aligned}
A_1&=X_1\cap S\cap T,& A_2&=X_1\cap(S\setminus T),\\
A_3&=X_1\cap(T\setminus S),& A_4&=X_1\setminus(S\cup T),
\end{aligned}
\]
and define $B_1,\ldots,B_4$ in the same way inside $X_2$. For a cell $U$, write $x_U=|U|/N$. For cells $U,V$, let $e_{UV}$ be the number of edges between nodes in $U$ and nodes in $V$, divided by $N^2$; an edge inside one cell is counted once. Thus $e_{UV}\ge0$, and $e_{UV}=0$ if either cell is empty. The four sets
$X_1,X_2,S,T$ each have normalized size $1/2$, that is,
\[
\sum_{U\subseteq Q}x_U=\frac12\qquad(Q\in\{X_1,X_2,S,T\}).
\]
These size equations, nonnegativity, and the empty-cell rule are all the
basic constraints on the normalized cell sizes and edge counts that we
need.

Next, we present the blocking and cut constraints we use. Let $I_Q(U,W)$
be $1$ if cells $U,W$ lie on the same side of $(Q,V\setminus Q)$, and $0$
otherwise. For cells $U\subseteq S$ and $U\subseteq T$, respectively, define
\begin{align*}
F_U&=2(1-\alpha)e_{UU}+\sum_{W\ne U}e_{UW}
 \bigl(\mathbf 1[W\subseteq S]-\alpha I_{X_1}(U,W)\bigr),\\
G_U&=2(1-\alpha)e_{UU}+\sum_{W\ne U}e_{UW}
 \bigl(\mathbf 1[W\subseteq T]-\alpha I_S(U,W)\bigr).
\end{align*}
The factor $2$ counts both endpoints of each edge inside $U$.
Both are nonnegative: multiplying either one by $N^2$ sums the appropriate
blocking gain over the cell. Moreover, $F_U+G_U>0$ for the cell
$U=S\cap T$, because it contains a real node. Write
\[
\sigma_{UW}=2I_{X_1}(U,W)-1,
\qquad D_U=2e_{UU}+\sum_{W\ne U}\sigma_{UW}e_{UW}.
\]
Averaging the change in the cut over swaps of one node from a cell
$U\subseteq X_1$ and one from $V\subseteq X_2$ gives the necessary row
\begin{equation}\label{eq:improved-local-swap}
L_{UV}(z):=x_VD_U+x_UD_V+2z e_{UV}\ge0\qquad(z\ge1/N).
\end{equation}
The $2z e_{UV}$ term corrects the swapped edge, which is counted twice in
$D_U+D_V$. For movement fractions $0\le q_U\le1$, it is convenient to
write the independent-movement cut change as
\[
\mathcal C(q)=\sum_{U\le V}\sigma_{UV}
 (q_U+q_V-2q_Uq_V)e_{UV}.
\]
Sampling a fixed number of nodes from a cell $H$ instead of moving
them independently adds exactly
$2q_H(1-q_H)e_{HH}/(|H|-1)$ to this expression when $|H|>1$.

There are two possible locations of the overlap. If $S\cap T=B_1$, put
$a=2x_{A_2}$ and $b=2x_{B_1}$. Then $x_{A_1}=0$ and the size equations
give $2x_{B_2}=1-a-b$; the other cell sizes do not enter the rows below.
Our choice of $X_1$ gives $a\ge1/2-\zeta$; this allows the isolate to
shift the original majority by one node. Also $b>0$ and $a+b\le1$.
In this case the mathematical
program consists of nonnegative $e_{UV}$ with the empty-cell rule,
$F_{A_2},F_{B_1},F_{B_2}\ge0$, the comparison
$E(S,V\setminus S)\ge E(X_1,X_2)$, and
$L_{A_2B_1}(\zeta),L_{A_2B_2}(\zeta)\ge0$.
The exact six-row dual certified in Appendix~\ref{apdx:core-certificate} shows that these constraints force $F_{B_1}+G_{B_1}\le0$, a contradiction.

Now suppose $S\cap T=A_1$, so $B_1$ is empty. Put
$b=2x_{A_1}$, $c=2x_{B_2}$, and $r=2x_{A_4}$. Every cell size is then
fixed by
\[
\begin{array}{c|ccccccc}
U&A_1&A_2&A_3&A_4&B_2&B_3&B_4\\\hline
2x_U&b&1-c-b&c-r&r&c&1-c-b+r&b-r.
\end{array}
\]
The overlap makes $b>0$; also $c>0$, since otherwise $S=X_1$ could not
block $X$. Our majority choice gives $c\le1/2$, and nonnegative cell
sizes give $b+c\le1$ and $0\le r\le\min\{b,c\}$. We retain the six blocking rows
\[
F_{A_1},F_{A_2},F_{B_2},G_{A_1},G_{A_3},G_{B_3}\ge0,
\qquad F_{A_1}+G_{A_1}>0,
\]
and four local rows $L_{A_1B_2}(\theta),L_{A_1B_3}(\theta),
L_{A_2B_2}(\theta),L_{A_2B_3}(\theta)\ge0$, where
\[
\theta=\min\{\zeta,b/2\}\ge1/N.
\]
Increasing the allowance in a local or coordinated row only weakens it;
the verifier may therefore use either $\zeta$ or $b/2$ in place of
$\theta$.
The strict overlap inequality also makes the sum of the six blocking
rows strictly positive.

When $0<b<c$ and $|A_2|\ge17$, we add one coordinated-exchange row. Set
$q_{A_1}=q_{B_2}=1$, all other fractions to zero except
\[
q_{A_2}=\frac{(c-b)(1/2-c)}{1-c-b},\qquad
q_{A_3}=\frac{(c-b)(1/2+c)}{c-r}.
\]
These fractions belong to $[0,1]$ and satisfy
\[
|A_2|q_{A_2}+|A_3|q_{A_3}=|B_2|-|A_1|.
\]
Choose $K$ from the floor and ceiling of $|A_2|q_{A_2}$ so that
$\E[K]=|A_2|q_{A_2}$. Conditional on $K$, independently choose uniformly
random subsets of sizes $K$ from $A_2$ and $|B_2|-|A_1|-K$ from $A_3$.
The mean is feasible, and the feasible values of $K$ form an interval
with integer endpoints, so both rounding choices respect the cell sizes. Together with moving
all of $A_1$ and $B_2$, this preserves balance in every realization.
The sampling corrections derived in Appendix~\ref{apdx:core-certificate}
then give the necessary inequality
\begin{equation}\label{eq:improved-coordinated}
\begin{aligned}
\mathcal C(q)+\theta\left(\frac{17q_{A_2}(1-q_{A_2})e_{A_2A_2}}{8x_{A_2}}+\frac{e_{A_3A_3}}{x_{A_3}}\right)+\frac{\theta^2e_{A_2A_3}}{2x_{A_2}x_{A_3}}\ge0.
\end{aligned}
\end{equation}
The first two correction terms bound sampling without replacement inside $A_2,A_3$, and the last bounds the interpolation error. Appendix~\ref{apdx:core-certificate} derives all three correction terms and verifies $|A_2|\ge17$ in both ranges where this row is used. We have $x_{A_2},x_{A_3}>0$ when $b<c$. For $0<b/c<1144/1575$, the entire program consists of the nonnegative edge variables and empty-cell rule, the six blocking rows, the four local rows, and \eqref{eq:improved-coordinated}.

We now show that the programs are infeasible. If $|B_2|\le2$, all nodes
of $B_2$ are real. The first blocking condition gives, for $i\in B_2$,
\[
d_i<(2-\alpha)\deg_{B_2}(i)
 -(\alpha-1)\deg_{X_2\setminus S}(i)
 -\deg_{X_1\setminus S}(i),
\]
where $d_i=2u_i(X)-\deg(i)$. Hence $d_i\le-1$ if $|B_2|=1$, and
$d_i\le0$ if $|B_2|=2$. If a real $j\in A_1$ were not adjacent to $i$, swapping them would reduce the cut. Thus $j$ is adjacent to every
node of $B_2$. These edges lie outside $T$, so
$u_j(T)\le2u_j(X)-d_j-|B_2|$. Since
$u_j(T)>\alpha^2u_j(X)\ge2u_j(X)$, integrality gives
$d_j\le-|B_2|-1$. Even an adjacent swap of $j$ and $i$ now reduces the
cut, a contradiction.

For the remaining ranges, a compact exact certificate completes the
argument. If $b\ge c$ and $|B_2|\ge3$, move all of $B_2$ and the same
number of nodes sampled from $A_1$. In this case take $z=c/6$ in
the local-swap row; $|B_2|\ge3$ ensures $z\ge1/N$. If $b<c$ and
$b/c\ge1144/1575$ with $|B_2|\ge5$, move all of $A_1$ and sample the
same number from $B_2$, taking $z=c/(2|B_2|)=1/N$ in the local row. In either case
the reverse-exchange constraint is
$\mathcal C(q)+2q_H(1-q_H)e_{HH}/(|H|-1)\ge0$. In the first case
$H=A_1$, $q_{A_1}=c/b$, $q_{B_2}=1$; in the second,
$H=B_2$, $q_{A_1}=1$, $q_{B_2}=b/c$. All other fractions are zero.
For $b\ge c$, the certificate uses the weaker bound
$8zc(b-c)e_{A_1A_1}/b^3$ in place of the exact correction; it follows
from $|A_1|\ge3$ and $z\ge1/N$. In each case, this exchange row,
$L_{A_1B_2}(z)$, and $F_{A_1},F_{B_2}\ge0$ force
$F_{A_1}+G_{A_1}\le0$ by symbolic
identities. The only high-ratio case with $|B_2|<5$ not handled above is
$(|B_2|,|A_1|)=(4,3)$. Here $c=8/N$. Taking the sharper valid allowance
$z=1/N=c/8$ in both the local rows and the coordinated exchange yields
the contradiction.

Finally, for $0<b/c<1144/1575$, set $\beta=b/c$ and $\rho=r/b$. The exact certificate described in Appendix~\ref{apdx:core-certificate} covers the entire rectangular domain $0<c\le1/2$, $0<\beta\le1144/1575$, $0\le\rho\le1$. It forces either $F_{A_1}+G_{A_1}\le0$ or the sum of the six blocking rows to be nonpositive, contradicting strict blocking. The next section describes the certificate and its independent checker for all the cases above. Hence either $X$ or $(S,V\setminus S)$ belongs to the $(38/25,0)$-core.

\section{Exact Certificate for Theorem~\ref{thm:two-core-improved}}\label{apdx:core-certificate}

The computational step in Theorem~\ref{thm:two-core-improved} is verified using a data file, \texttt{certificate.json}, and a 128-line Python checker, \texttt{check.py}.\footnote{All the code is available at \url{https://github.com/vignesh-viswanathan/improved-fair-graph-partitions}.} We use the notation of its proof, with $N\ge240$ even and $\zeta=1/240$.

\paragraph{Sampling correction.}
We first justify the $A_2$ correction in \eqref{eq:improved-coordinated}. Write $m=|A_2|$. In the range $b/c\le1144/1575$, the bound $c\le1/2$ gives
\begin{align*}
m=\frac{N(1-c-b)}2\ge\frac{431N}{6300}\ge\frac{1724}{105}>16.
\end{align*}
Thus $m\ge17$ by integrality. In the exceptional case $(|B_2|,|A_1|)=(4,3)$, we have $m=N/2-7\ge113$, so the same bound holds.

For the coupled sample count $K$ defined above, $\E[K]=mq_{A_2}$. The correction to independent movement for edges within $A_2$ is
\begin{align*}
\left(\frac{2q_{A_2}(1-q_{A_2})}{m-1}-\frac{2\operatorname{Var}(K)}{m(m-1)}\right)e_{A_2A_2}.
\end{align*}
Indeed, the probability that an internal edge is cut is $\E[2K(m-K)/(m(m-1))]$, and $\E[K(m-K)]=m^2q_{A_2}(1-q_{A_2})-\operatorname{Var}(K)$. Since $m/(m-1)\le17/16$, $m=Nx_{A_2}$, and $z\ge1/N$, the correction coefficient is at most
\begin{align*}
\frac{2q_{A_2}(1-q_{A_2})}{m-1}&\le\frac{17q_{A_2}(1-q_{A_2})}{8m}\\
&\le\frac{17zq_{A_2}(1-q_{A_2})}{8x_{A_2}}.
\end{align*}
For $|A_3|\ge2$, the same calculation bounds the correction coefficient
for edges within $A_3$ by
\begin{align*}
\frac{2q_{A_3}(1-q_{A_3})}{|A_3|-1}
&\le\frac{1}{2(|A_3|-1)}\\
&\le\frac{1}{|A_3|}\le\frac{z}{x_{A_3}}.
\end{align*}
If $|A_3|=1$, there are no internal edges and hence no correction.

The sample counts in $A_2$ and $A_3$ sum to $|B_2|-|A_1|$, so their
covariance is $-\operatorname{Var}(K)$. Conditional independence of the
two subsets therefore makes the probability that both endpoints of an
$A_2A_3$ edge move equal to
\[
q_{A_2}q_{A_3}-\frac{\operatorname{Var}(K)}{|A_2||A_3|}.
\]
Both cells lie in $X_1$, so the correction for these edges relative to
$\mathcal C(q)$ is bounded by
\begin{align*}
\frac{2\operatorname{Var}(K)e_{A_2A_3}}{|A_2||A_3|}
&\le\frac{e_{A_2A_3}}{2|A_2||A_3|}\\
&\le\frac{z^2e_{A_2A_3}}{2x_{A_2}x_{A_3}},
\end{align*}
where $\operatorname{Var}(K)\le1/4$ because $K$ is supported on two
adjacent integers. All other edge types have the same expected
contribution as in $\mathcal C(q)$. Taking $z=\theta$ and using the
nonnegative expected cut change of this balanced exchange proves
\eqref{eq:improved-coordinated}.

\paragraph{Certificate.}
For each parameter range, let $L_i\ge0$ be the necessary blocking and cut inequalities used in the main proof. Each certificate identity has the form
\begin{align*}
-J=\sum_i\lambda_i L_i+\sum_{U\le V}\mu_{UV}e_{UV},\qquad \lambda_i,\mu_{UV}\ge0,
\end{align*}
where the coefficients are rational functions of the cell-size parameters. Here $J$ is the combined blocking gain $F_U+G_U$ of the overlap cell $U$, or the sum of the six blocking gains when the overlap is $A_1$. In either case strict blocking gives $J>0$, whereas the identity forces $J\le0$.

The certificate contains 12 such identities and a cover of the following five parameter domains. For $b<c$, write $\beta=b/c$ and $\rho=r/b$.
\begin{enumerate}
\item If the overlap is $B_1$, the domain is $1/2-\zeta\le a\le1$, $0<b\le1/2+\zeta$, with allowance $z=\zeta$.
\item If the overlap is $A_1$, $b\ge c$, and $|B_2|\ge3$, the domain is $0<c/b\le1$, with allowance $z=c/6$.
\item If $b<c$ and $|B_2|\ge5$, the high-ratio domain is $1144/1575\le\beta\le1$. The certificate uses the weaker allowance $z=c/10$ and reverse-exchange correction $\beta(1-\beta)e_{B_2B_2}/2$, both valid because $|B_2|\ge5$.
\item If $(|B_2|,|A_1|)=(4,3)$, the domain is $0<c\le8\zeta$, $0\le\rho\le1$, with $b=3c/4$ and allowance $z=c/8$.
\item The low-ratio domain is $0<c\le1/2$, $0<\beta\le1144/1575$, $0\le\rho\le1$, with either $z=\zeta$ or $z=b/2$.
\end{enumerate}
The first four domains each use one box, and the last uses 153, for a total of 157 boxes. Together with the exclusion of $|B_2|\le2$ in the main proof, these domains cover every possible pair of blocking sets.

\paragraph{Checker.}
With the two files in the same directory and SymPy installed, running \texttt{python3 check.py} verifies the certificate using exact integer and rational arithmetic. The checker reconstructs the row coefficients from the formulas in the proof, checks each identity symbolically, and proves the coefficient signs using rational Bernstein coefficients on each box. It also verifies that every denominator stays nonzero, including on included boundary faces. The boxes are specified by a binary tree that bisects a coordinate interval at every internal node and includes both children; this verifies complete coverage of each domain. The certificate is read only as data, and checking it requires no optimization or floating-point tolerance. Thus the correctness of Theorem~\ref{thm:two-core-improved} depends only on the mathematical reduction above and this small checker, regardless of how the certificate was found.

\end{document}